\documentclass[11pt]{article}
\usepackage[utf8]{inputenc}
\usepackage{fullpage}
\usepackage{amsmath,amsthm,amssymb, color, tikz, mathtools}
\usepackage{tcolorbox}
\usepackage{mdframed}

\usepackage{amsfonts}
\usepackage{amsbsy}
\usetikzlibrary{positioning}
\usetikzlibrary{shapes}
\usetikzlibrary{decorations.pathreplacing}

\usepackage{thm-restate}
\usepackage[hidelinks]{hyperref}
\usepackage{cleveref}

\title{Approximate Degree Composition for Recursive Functions}

\newcommand{\paren}[1]{\left( #1 \right)}
\newcommand{\brac}[1]{\left[ #1 \right]}
\newcommand{\set}[1]{\left\{ #1 \right\}}
\newcommand{\parencond}[2]{\left( #1 \;\middle\vert\; #2 \right)}
\newcommand{\braccond}[2]{\left[ #1 \;\middle\vert\; #2 \right]}
\newcommand{\setcond}[2]{\left\{ #1 \;\middle\vert\; #2 \right\}}

\newcommand{\abs}[1]{\left\lvert #1 \right\rvert}
\newcommand{\norm}[1]{\left\lVert #1 \right\rVert}
\newcommand{\ang}[1]{\left\langle #1 \right\rangle}
\newcommand{\floor}[1]{\left\lfloor #1 \right\rfloor}
\newcommand{\ceil}[1]{\left\lceil #1 \right\rceil}
\newcommand{\blue}[1]{\textcolor{blue}{#1}}

\newcommand{\wt}{\widetilde}
\newcommand{\wh}{\widehat}
\newcommand{\ol}{\overline}
\newcommand{\tn}[1]{\textnormal{#1}}

\newcommand{\MAJ}{\mathsf{MAJ}}
\newcommand{\polylog}{\text{polylog}}

\DeclareMathOperator*{\Exp}{\mathbb{E}}

\usepackage{bbm}
\newcommand{\indicator}{\mathbbm{1}}

\newcommand{\zone}{\set{0,1}}
\newcommand{\zonep}{\MAJset{0,1,*}}
\newcommand{\pmone}{\set{-1,1}}

\newcommand{\ftwo}{\mathbb{F}_2}
\newcommand{\reals}{\ensuremath{\mathbb{R}}}

\newcommand{\OR}{\mathsf{OR}}
\newcommand{\PrOR}{\mathsf{PrOR}}
\newcommand{\AND}{\mathsf{AND}}
\newcommand{\XOR}{\mathsf{XOR}}
\newcommand{\PARITY}{\mathsf{PARITY}}
\newcommand{\PAR}{\mathsf{PARITY}/\neg\mathsf{PARITY}}
\newcommand{\rub}{\mathsf{RUB}}
\newcommand{\sink}{\mathsf{SINK}}
\newcommand{\Maj}{\mathsf{MAJ}}

\newcommand{\boolfn}[1]{\ensuremath {\zone}^{#1} \to \zone}

\renewcommand{\tilde}{\widetilde}
\renewcommand{\Tilde}{\widetilde}
\newcommand{\adeg}{\widetilde{\deg}}

\newcommand{\odeg}{\widetilde{\mathrm{odeg}}}
\newcommand{\rqc}{\textnormal{R}}
\newcommand{\s}{\textnormal{s}}
\newcommand{\bs}{\textnormal{bs}}

\newcommand{\fbs}{\textnormal{fbs}}
\newcommand{\supp}{\textnormal{supp}}
\newcommand{\spann}[1]{\textnormal{span}\left\{#1\right\}}

\newcommand{\sign}{\textnormal{sign}}

\newtheorem{theorem}{Theorem}[section]
\newtheorem{corollary}[theorem]{Corollary}
\newtheorem{lemma}[theorem]{Lemma}
\newtheorem{claim}[theorem]{Claim}
\newtheorem{defi}[theorem]{Definition}

\newtheorem{observation}[theorem]{Observation}

\newtheorem{fact}[theorem]{Fact}

\newtheorem{open question}[theorem]{Open question}

\newcommand{\Dom}{\textnormal{Dom}}

\newcommand{\amp}{\mathsf{Amp}}
\date{}
\author{
Sourav Chakraborty\thanks{Indian Statistical Institute, Kolkata, India. \texttt{sourav@isical.ac.in}}
\and 
Chandrima Kayal\thanks{Indian Statistical Institute, Kolkata, India. \texttt{chandrimakayal2012@gmail.com}}
\and 
Rajat Mittal\thanks{Indian Institute of Technology Kanpur, India.  \texttt{rmittal@cse.iitk.ac.in }}
\and
Manaswi Paraashar\thanks{University of Copenhagen, Denmark. \texttt{manaswi.isi@gmail.com}}
\and
Nitin Saurabh \thanks{Indian Institute of Technology Hyderabad, India. \texttt{nitin@cse.iith.ac.in}}
}
\begin{document}

\maketitle
\begin{abstract}
Determining the approximate degree composition for Boolean functions remains a significant unsolved problem in Boolean function complexity. In recent decades, researchers have concentrated on proving that approximate degree composes for special types of inner and outer functions.
An important and extensively studied class of functions are the recursive functions, i.e.~functions obtained by composing a base function with itself a number of times. 
    The main result of this work is to show that the approximate degree composes if either of the following conditions holds:
    \begin{itemize}
        \item The outer function $f:\{0,1\}^n\to \{0,1\}$ is a recursive function of the form $h^d$, with $h$ being any base function and $d= \Omega(\log\log n)$.
        \item The inner function is a recursive function of the form $h^d$, with $h$ being any constant arity base function (other than AND and OR) and $d= \Omega(\log\log n)$, where $n$ is the arity of the outer function.
    \end{itemize}

    In terms of proof techniques, we first observe that the lower bound for composition can be obtained by introducing majority in between the inner and the outer functions. We then show that majority can be \emph{efficiently eliminated} if the inner or outer function is a recursive function.
\end{abstract}



 \section{Introduction}

Representations of Boolean functions $f\colon \boolfn{n}$ in terms of multivariate polynomials $p(x)$ play a pivotal role in theoretical computer science. There are different notions of representations;
\begin{itemize}
\item \emph{exact} representation -- $f(x) = p(x)$ for all $x\in\zone^n$,
\item \emph{approximate} representation -- $|f(x)-p(x)| \leq 1/3$ for all $x\in\zone^n$, and 
\item \emph{sign} representation -- $(1-2f(x))p(x) > 0$ for all $x\in\zone^n$.
\end{itemize}

Arguably the most important measure associated with a polynomial is its (total) \emph{degree}. Let $\deg(f)$, $\adeg(f)$, and $\deg_{\pm}(f)$ denote the minimal possible degree of a real polynomial \emph{exactly}, \emph{approximately}, and \emph{sign} representing $f$, respectively. These different notions of degrees capture notions of efficiency in many different models of computation (e.g., decision trees, quantum query, perceptrons), 
and are thus well-studied in literature (see, e.g., \cite{Beigel93,Beigel94a,BunT22} and the references therein).

For instance, $\deg_{\pm}(f)$ (called sign degree) has strong connections to -- separations among complexity classes \cite{Beigel94a},  designing efficient learning algorithm \cite{KlivansS04,KlivansOS04}, and lower bounds against circuits, formulas, communication complexity, etc.~\cite{BuhrmanVW07,Chen16}. Similarly, upper bounds on $\adeg(f)$ (called approximate degree), has strong connections to learning theory \cite{KalaiKMS08,KlivansS06,ServedioTT12}, approximate inclusion-exclusion \cite{KahnLS96,Sherstov09}, differentially private data release \cite{ThalerUV12,ChandrasekaranTUW14}, etc. While the lower bounds on approximate degree lead to lower bounds in quantum query complexity \cite{BealsBCMW01,AaronsonS04,Aaronson12}, communication complexity \cite{Sherstov09, Sherstovpattern11}, circuit complexity \cite{Allender89}, etc. 

Despite decades of work in this area, there are many important problems that are yet to be resolved completely. One such problem pertains to the composition of approximate degrees. For any two Boolean functions $f\colon\boolfn{n}$ and $g\colon\boolfn{m}$, define the \emph{composed} function $f\circ g\colon \boolfn{nm}$ as follows 
\begin{align*}
    f \circ g (x_{11},\ldots ,x_{1m},\ldots\ldots ,x_{n1},\ldots ,x_{nm}) = f(g(x_1),\ldots ,g(x_n)),
\end{align*}
where $x_i = (x_{i1},\ldots ,x_{im}) \in \zone^m$ for $i \in [n]$. The function $f$ is called the outer function and $g$ the inner function.  

Investigating the behaviour of complexity measures under composition has been a quintessential tool in our quest to gain insights into relationships among different measures. In particular, composition has been used successfully on numerous occasions to show separations between various complexity measures associated with Boolean functions, see, e.g.,  \cite{SaksW86,NS94,JuknaRSW99,Ambainis05,Tal13,GSS16}.  
A big open problem in this context is to understand how approximate degree behaves under composition. More formally, it asks whether for all Boolean functions $f:\{0,1\}^n \to \{0,1\}$ and $g:\{0,1\}^m\to \{0,1\}$, 
    \begin{align*}
        \adeg (f \circ g) = \Tilde{\Theta}(\adeg(f) \cdot \adeg(g))?
    \end{align*}
The tilde in the $\Tilde{\Theta}$ notation hides a factor polynomial in $\log(n + m)$. This problem is often referred to as the `approximate degree composition' problem.

The upper bound, $\adeg(f\circ g) = O(\adeg(f)\cdot\adeg(g))$, was established in a seminal work \cite{robustSherstov13} of Sherstov. 
Thus to completely resolve the problem it remains to prove a matching lower bound on the approximate degree of a composed function in terms of the approximate degree of the individual functions. In other words, does the following hold for all Boolean functions $f$ and $g$,  
\begin{align*}
 \adeg(f\circ g) = \tilde{\Omega}\left(\adeg(f)\cdot\adeg(g)\right)?
\end{align*}
In this article we will refer to the aforementioned (lower bound) question by the phrase `approximate degree composition' problem.  

Numerous works, including those by \cite{NS94,Ambainis05,Sherstov12,She13,She13a,BT13,DavidBGK18,CKMPSS23}, actively pursued these lower bounds, leading to newer connections with several important problems in the field. 
However, establishing the lower bound $\adeg(f\circ g) = \tilde{\Omega}\left(\adeg(f)\;\adeg(g)\right)$ even for specific functions or restricted classes of functions is often very challenging. 
For example, consider the composed function $\OR\circ\AND$, it took a long series of work \cite{NS94,Shi02,Ambainis05,She13,She13a,BT13} over nearly two decades to prove that $\adeg(\OR\circ \AND) = \Omega\left(\adeg(\OR)\;\adeg(\AND)\right)$. Till date we know that the approximate degree composes in the following cases: 
\begin{itemize}
    \item when the outer function $f$ has full approximate degree, i.e., $\Theta(n)$ \cite{Sherstov12}, 
    \item when the outer function $f$ is a symmetric function \cite{DavidBGK18}, 
    \item when the outer function $f$ has minimal approximate degree with respect to its block sensitivity, i.e., $\adeg(f) = O(\sqrt{\bs(f)})$ \cite{CKMPSS23}, and
    \item when the sign degree of the inner function is same as its approximate degree~\cite{Sherstov12, Lee09}.
\end{itemize}

This work focuses on the behavior of approximate degree when \emph{recursive functions} are composed with other general functions (as outer or inner function). Here, by recursive functions, we mean the functions of the kind $h^d$ ($h$ composed with itself $d$ times) where the arity of $h$ is small. The function $h$ is often called the base function and the function $f$ is called the recursive-$h$ function. 

Recursive functions are an important class of Boolean functions that are studied in various different contexts in the analysis of Boolean functions, mainly in proving various lower bounds ~\cite{Ambainis05,Snir85,SW86,NS94,NW94,DHT17}.  For example, the Kushilevitz's function~\cite{NW94} which is the only known non-trivial example of functions with low degree and high sensitivity is a recursive function of a carefully chosen base function.  Recursive majority, $\Maj_3^d$, is another recursive function that has been studied extensively in the literature for its different properties 
~\cite{SaksW86,JayramKS03,Leonardos13,MagniezNSSTX16}. Boppana (see, e.g., \cite{SaksW86}) used it to provide the first evidence that the randomized query is more powerful than deterministic query~\cite{SaksW86}. In the same article, they show a similar separation using recursive $\AND_2 \circ \OR_2$ function too.
In a different application of recursive $\AND_2 \circ \OR_2$,~\cite{JuknaRSW99} show separation between deterministic tree-size complexity and number of monomials in the minimal DNF or CNF.



The approximate degree composition was not known when the outer or inner function is a recursive function, in general. For some special recursive functions, however, it was known that the approximate degree composes. For example, the $\OR$ function on $n=3^d$ bits is same as $\OR_3^d$. After a series of works (\cite{NS94,Ambainis05,She13a,BT13,She13}), it was proven that the approximate degree composition holds when the outer function is $\OR$, and in general symmetric~\cite{DavidBGK18}.  Similarly, from the result of \cite{Sherstov12,Lee09} it can be observed that the lower bound holds when either the inner or outer function is recursive $\PARITY$. Unfortunately, these results can't be applied in general even when the base function is symmetric or it has full approximate degree.

 
This scenario leads to the natural question: 
\begin{quote}
\begin{center}
    \emph{Can we prove that $\adeg(f \circ g) = \Omega(\adeg(f) \cdot \adeg(g))$ when the outer function $f$ or the inner function $g$ is recursive?}
\end{center}
\end{quote}

\subsection{Our Results}

Let $h : \zone^k \to \zone$ be a function on $k$-bits. Let $h^d$ denotes the complete $k$-ary tree of depth $d$ such that each internal node of the tree is labelled by $h$ and the leaves of tree are labelled by distinct variables. Our main result shows that the composition theorem holds for any $h^d$ (except a few specific $h$'s), either as the outer function with any inner function or as the inner function with any outer function. 
\begin{restatable}{theorem}{anyrecursivefunction}
    \label{th:for any recursive function}
    Let $f\colon\boolfn{n}$ and $g\colon\boolfn{m}$ be two Boolean functions and $d=\Omega(\log\log n)$. Then, 
    \[\adeg(f\circ g) = \Omega\left(\frac{\adeg(f)\adeg(g)}{\textup{\polylog}(n)}\right),\] if either of the following conditions hold:
    \begin{enumerate}
    \item $f= h^d$, for any Boolean function $h$.  
    \item $g= h^d$, for any Boolean function $h$ with constant arity and not equal to $\AND$ or $\OR$. 
    \end{enumerate}
\end{restatable}


In light of the above theorem,
understanding the composition of approximate degree when inner function is $\OR$ is the central case for making progress towards the general composition question.  


We would like to emphasize that there are not many results which prove composition theorem for a general class of inner functions. Theorem~\ref{th:for any recursive function} shows that the composition property holds if the inner function is recursive irrespective of the outer function. 


We further note that Theorem~\ref{th:for any recursive function} doesn't follow from the known results even when the composition theorem is known to hold for the base function.   
Firstly, it is known that the composition lower bound  holds when the outer function is symmetric~\cite{DavidBGK18}; though, a repeated composition of a symmetric function will incur the factor of $(\log n)^d$ (because of the $\log n$ factor hiding in the $\tilde{\Omega}$ notation).  
Secondly, while the majority function, $\Maj_n$, has full approximate degree ($\Theta(n)$), $\Maj_3^d$ doesn't have full approximate degree.  Thus, Sherstov's result \cite{Sherstov12} that proves composition theorem holds for functions with full approximate degree cannot be applied in the case of recursive majority. 
The situation is similar for the inner function as well.

Moving ahead, the proof of Theorem~\ref{th:for any recursive function} uses two ideas.
\begin{itemize}
    \item We first prove that a similar theorem works for the specific case of $h=\MAJ_3$ and $h=\AND_2 \circ \OR_2$ functions.
    \item Then, we use a general $h$ to \emph{simulate} $\AND_2 \circ \OR_2$; hence, proving composition for the general case.
\end{itemize}

The case of recursive $h=\MAJ_3$ and $h=\AND_2 \circ \OR_2$ functions is in itself very interesting. There have been several works towards exploring the approximate degree and other properties of these two functions~\cite{GoosJ16,KV14, SaksW86, JuknaRSW99}. Given their importance, and the fact that it is a central step in our main result (Theorem~\ref{th:for any recursive function}), we state the composition theorem for these two functions separately.


\begin{restatable}{theorem}{RecursiveMajorityMainThm}
\label{th: RecMaj composition}
Let $f$ and $g$ be two Boolean functions. Then, 
\begin{align*}
    \adeg(f \circ h^d) = \Tilde{\Omega}(\adeg(f)~\adeg(h^d)) \text{ and }\;
    \adeg(h^d \circ g) = \Tilde{\Omega}(\adeg(h^d)~\adeg(g)),
\end{align*}
where $h$ is either $\MAJ_3 : \zone^3 \to \zone$ or $\AND_2 \circ \OR_2 : \zone^4 \to \zone$, $n$ is the arity of the outer function, $d \geq C \log\log n$ for a large enough constant $C$, and 
$\Tilde{\Omega}(\cdot)$ hides $\textup{\polylog}(n)$ factors.
\end{restatable}

%

To prove Theorem~\ref{th: RecMaj composition} we will need the following lemma. Even though the lemma can be obtained from a combination of known results (e.g., \cite{Sherstov12} and \cite{BoulandCHTV17}) with appropriate parameters, we give a self-contained simpler proof of the lemma, inspired by the primal-dual perspective of \cite{She13a}.

\begin{restatable}{lemma}{MajorityHardnessamplification} \label{lem: Maj as amplification function}
     For any Boolean function $f:\zone^n \to \zone$ and $g:\zone^m \to \zone$,
     \begin{align}
     \label{eq: maj comp}
         \adeg(f\circ \mathsf{MAJ}_t \circ g) = \Omega( \adeg(f) \adeg(g))
     \end{align} for $t=\Omega(\log n)$. 
\end{restatable}

Note that, Lemma~\ref{lem: Maj as amplification function} gives a way to settle the composition question affirmatively. 
In particular, if $\adeg(f \circ \Maj_t \circ g) = \Tilde{O}(\adeg(f \circ g))$, where $t$ is $\Theta(\log n)$ and $n$ is the arity of $f$, then it follows that the composition holds for $f$ and $g$.



We also highlight that a tighter lower bound can be obtained when the middle function $\Maj$ is replaced by an ``amplifier function" in Lemma~\ref{lem: Maj as amplification function}. 
Define $H$ to be a \emph{strong hardness amplifier function} for $g$ if \[\adeg_{\frac{1-2^{-\Omega(t)}}{2}}(H \circ g) = \Omega(\adeg(H) \circ \adeg(g)).\] In Lemma~\ref{lem: full degree function as hardness amplifier function } we observe $\adeg(f\circ H \circ g) = \Omega( \adeg(f) \adeg(H) \adeg(g))$, when $H$ is a strong hardness amplifier function for $g$. We discuss this improvement in Appendix~\ref{appendix:composition theorem for hardness amplifier}.

\subsection{Proof Ideas}

To address the lower bound for the composition of two Boolean functions $f$ and $g$, $f \circ g$, we will call $f$ to be the `outer function' and g to be the `inner function'. In the case of three-layer composed function $(f \circ H \circ g)$, we will call $H$ to be the `hardness amplifier' and $f$ and $g$ to be the outer and inner functions respectively. \\

\emph{Primal dual approach to composition:} \\
Our proof technique is based on the primal-dual view used by~\cite{She13a} for proving the composition of $\AND_n \circ \OR_n$. Here, instead of using `dual-composition method' (see \cite{BT13,BunT22}) we will be using only the dual witness of the inner function. The primal-dual approach is to construct an approximating polynomial for $f$ with smaller degree than $\adeg(f)$ by applying a linear operator $L$ on the assumed approximating polynomial for $f \circ g$ (say $p$, with smaller degree than claimed), leading to a contradiction. The linear operator $L$ is defined by taking the input to $f$, extending it to a probability distribution (which depends upon the dual of $g$) over the inputs of $f \circ g$ and outputting the expectation.

Let $\psi$ be the dual witness of $g$, we get $\mu_0$ and $\mu_1$ by restricting $\psi$ on support which takes positive and negative values respectively; by the properties of dual witness, $\mu_1$ (and $\mu_0$) will mostly be supported on inputs $x$ such that $g(x) = 1$ (and $g(x)=0$ respectively). 
The input to $f$ is expanded bit by bit using $\mu_0$ and $\mu_1$, creating a distribution on inputs of $f\circ g$. 

Formally, $L$ takes a general function $h:\zone^{mn} \to \zone$ and gives $Lh:\zone^n \to \reals$. 
 
\begin{align}
         Lh(z_1, \dots, z_n) = \Exp_{x_1 \sim \mu_{z_1}}\Exp_{x_2 \sim \mu_{z_2}}\cdots \Exp_{x_n \sim \mu_{z_n}} [h(x_1,x_2,\ldots ,x_n)], 
\end{align} 
where $x_i \in \zone^m$ for all $i \in \{1,2,\dots,n\}$.

To complete the proof, the following two properties of $L$ are required:

\begin{enumerate}


\item Showing that the polynomial $Lp$ indeed approximates $f$ in $l_{\infty}$ norm. Intuitively this happens because the restricted distributions ($\mu_0$ and $\mu_1$) are a pretty good indicator of the value of $g$.

\item The degree of $Lp$ is small, intuitively because   
$L$ reduces the degree of every monomial by a factor of $\adeg(g)$. 
    

\end{enumerate}

\emph{Problem with the primal dual approach} \\
Unfortunately, the recipe described above doesn't work well in general due to the error introduced by the expectation over $\mu_0$ and $\mu_1$ in the string $(z_1, \dots, z_n)$. To handle a noisy string in place of a Boolean string, the approximating polynomial $p$ needs to be robust. 
A polynomial is robust to noise $\frac{1}{3}$, if for all inputs $x$ and for all $\Delta \in [-\frac{1}{3},\frac{1}{3}]^m$, $|p(x) - p(x+\Delta)| <\varepsilon$. 

While any polynomial $p$ can be made robust up to error $\epsilon$ with degree at most $\adeg(p)+\log(\frac{1}{\epsilon}))$ (see Theorem~\ref{thm:robustification } by \cite{robustSherstov13}), 
such polynomials are not known to be multilinear, making the analysis of expectation difficult. \cite{BuhrmanNRW07} gives a robust multilinear polynomial for any Boolean function $f:\zone^n \to \zone$; though, the polynomial is defined on a perturbation matrix of input $x$ instead of $x$ itself. We now discuss how to overcome this problem.  

We give the proof ideas of Theorem~\ref{th:for any recursive function}, Theorem~\ref{th: RecMaj composition} and Lemma~\ref{lem: Maj as amplification function} in the reverse order, the way they are obtained from each other. \\

\emph{Proof idea of Lemma~\ref{lem: Maj as amplification function}} \\
We will use $\Maj_t$ to get past this difficulty; it helps to reduce the noise in the input of $f$ to error $\frac{1}{n}$. Using the fact that any multilinear polynomial on $n$ variables is robust up to error $\frac{1}{n}$, we have our lower bound for the function $\adeg(f \circ \Maj_t \circ g)$ where $t = \Omega(\log n)$. \\       

        \emph{Proof idea of Theorem~\ref{th: RecMaj composition}} 
        \\
        Using previously known constructions (\cite{VALIANT1984363,Goldreich20}), $\Maj_{\log n}$ can be projected to $\Maj_3^d$ and $(\AND_2 \circ \OR_2)^d$, where $d \geq C \log \log n$. We now replace $\Maj_{\log n}$ in Lemma~\ref{lem: Maj as amplification function} with these recursive functions; by using the associativity of the composition of functions and the approximate degree upper bound~\cite{robustSherstov13}, we finish the proof of the theorem. Note that we only lose a factor of $\polylog(n)$ in the lower bound since we only need to simulate $\Maj_{\log n}$.  \\

Now we give the idea about how to replace $\AND_2 \circ \OR_2$ with almost any recursive function to get our main result. \\

        \emph{Proof idea of Theorem~\ref{th:for any recursive function}}\\
         Given Theorem~\ref{th: RecMaj composition}, it is natural to ask, what other recursive functions satisfy the composition property. We show that almost any $h$ can be used to replace the $\AND_2 \circ \OR_2$ function. This is done by simulating $\AND_2$ and $\OR_2$ using restrictions of $h$ and its powers.
         The proof of this simulation is divided into two cases: monotone $h$ and non-monotone $h$. 
         
         For the monotone case (except when $h$ is $\AND$ or $\OR$): We show that both $\AND_2$ and $\OR_2$ will be present as sub-cubes of the original Boolean hypercube of $h$. 

         For the non-monotone case (except when $h$ is $\PARITY$ or $\neg \PARITY$): The proof requires more work here because of these two issues. First, there need not be both functions $\AND_2$ and $\OR_2$ as sub-cubes (though, we show that at least one will be present). Second, the sub-cube could be rotated. The resolution to both these issues is same. We use the non-monotonicity to construct the negation function. This allows us to rotate the sub-cube as well as construct $\AND_2/\OR_2$ from the other one. 
          
        A slight technical point to note is that when $h$ is a non-constant arity function and $h^d$ is the inner function, then the loss in the lower bound will be larger than $\polylog(n)$. However,  even for the case when the base function $h$ has arity that is a ``slowly" growing function of $n$  
        we still obtain a non-trivial lower bound composition result.

         The remaining cases of Theorem~\ref{th:for any recursive function}, i.e., $(i)$ when $f$ or $g$ equals $h^d$ for $h \in\{\PARITY,\neg\PARITY \}$ follows from \cite{Sherstov12}, and
         $(ii)$ when $f=h^d$ and $h \in \{\AND,\OR\}$ follows from  \cite{DavidBGK18}.



\section{Notations and Preliminaries}

In this paper, we will assume a Boolean function has domain $\{0,1\}^n$ and range $\{0,1\}$. 
We start with some of the important definitions. 

\begin{defi}[Generalized Composition of functions]
\label{defi: Generalized composition of functions}
For any Boolean function $f:\zone^n \to \{0,1\}$ and $n$ Boolean functions $g_1, g_2, \dots, g_n$, define the composed function 
\[f \circ (g_1, g_2, \dots, g_n)(x_1,x_2, \dots, x_n)= f(g_1(x_1),g_2(x_2), \dots, g_n(x_n)),\]
where $g_i$'s can have different arities and $x_i \in \Dom(g_i)$ for all $i \in [n]$.

When all the copies of $g_i$ are the same function $g$ then the composed function is denoted by $(f \circ g)$. 
\end{defi}

\begin{defi}[Recursive functions]
\label{def: recursive functions}
For any Boolean function $f:\zone^t \to \zone$ we define recursive function $f^d : \zone^{t^d} \to \zone$ by $f^d= \underbrace{f \circ f \circ  \ldots \circ f}_{d \textit{ times }}.$
    \label{defi: }
\end{defi}
\begin{defi}[Approximate degree $(\adeg)$]
\label{defi: approximate degree}
For some constant $0 < \epsilon < 1$, a polynomial $p: \mathbb{R}^n \to \mathbb{R}$ is said to $\epsilon$-approximate a Boolean function $f: \zone^n \to \{0,1\}$ if 
 \(   | p(x) -f(x) | \leq \epsilon, \quad \forall x \in \zone^n\).
The approximate degree of $f$, $\adeg_{\epsilon}(f)$, is the minimum possible degree of a polynomial that $\epsilon$-approximates $f$. Conventionally we use $\adeg(\cdot)$ as the shorthand for $\adeg_{1/3}.(\cdot)$ 
\end{defi}
Note that the constant $\epsilon$ in the above definitions can be replaced by any constant strictly smaller than $1/2$ which changes $\adeg_{\epsilon}(f)$ by only a constant factor.

 We also need the following facts about error reduction in approximating polynomials. 
    \begin{lemma}[Error reduction~\cite{Sherstov11}]
        \label{lem:error-reduction-adeg}
        For any $\varepsilon > 0$, \(\adeg_{\varepsilon}(f) = \Theta_\varepsilon(\adeg(f))\), where $\Theta_\varepsilon(\cdot)$ denotes that the constant depends on $\varepsilon$.
    \end{lemma}

\begin{lemma}[\cite{Sherstovpattern11,Sherstov12}]
    \label{lem:adeg-dual}
    Let $f : \zone^n \to \reals$ be a function and $\varepsilon >0$. Then, $\adeg_{\varepsilon}(f) \geq d$ iff there exists a function $\psi : \zone^n \to \reals$ such that 
    \begin{align}
    \sum_{x \in \zone^n}|\psi(x)| & = 1, \label{eq:adeg-dual-l1} \\
    \sum_{x \in \zone^n}\psi(x)\cdot f(x) & > \varepsilon, \label{eq:adeg-dual-corr} \\
    \sum_{x \in \zone^n}\psi(x)\cdot p(x) & = 0 \label{eq:adeg-dual-0}
    \end{align}
    for every polynomial $p$ of degree $< d$. 
\end{lemma}

 In a seminal work, Sherstov~\cite{robustSherstov13} showed that approximate degree can increase at most multiplicatively under composition.
\begin{theorem}[\cite{robustSherstov13}]
    \label{thm:robust-comp}
    For all Boolean function $f:\zone^n \to \zone$ and $g:\zone^{m} \to \zone$, 
    \(\adeg(f\circ g) = O(\adeg(f)\cdot \adeg(g))\).
\end{theorem}
 We will be working with inputs that are not Boolean \emph{but} are close to Boolean. We need the notion of \emph{robust} approximating polynomials. 
    \begin{defi}[$(\delta,\varepsilon)$-robust approximating polynomial]
        Let $p:\zone^m \to \zone$ be a polynomial. Then, for $\delta, \varepsilon >0$, a $(\delta,\varepsilon)$-robust approximating polynomial for $p$ is a polynomial $p_{robust} : \reals^m \to \reals$ such that for all $x \in \zone^m$ and for all $\Delta \in [-\delta,\delta]^m$, 
        \[|p(x) - p_{robust}(x+\Delta)| <\varepsilon.\]
    \end{defi}
    
    Note that robust polynomial need not to be multilinear. \cite{robustSherstov13} proved that for any Boolean function $f: \zone^n \to \zone$ there exists a robust polynomial with degree at most $(\deg(p) + \log(1/\varepsilon))$.
    
    \begin{theorem}[Sherstov~\cite{robustSherstov13}]
        \label{thm:robustification }
        A $(\delta,\varepsilon)$-robust approximating polynomial for $p$ of degree $O_{\delta}(\deg(p) + \log(1/\varepsilon))$ exists. Here $O_\delta(\cdot)$ denotes that the constant in $O(\cdot)$ depends on $\delta$.
    \end{theorem}

    For our purposes, we need a multilinear robust approximating polynomial. 
 \begin{theorem}[Folklore]
    \label{folklorerobust}
        Any multilinear polynomial $p:\zone^n \to \zone$ is $(\frac{\delta}{n}, \delta)$-robust. 
    \end{theorem}



\begin{theorem}[\cite{DavidBGK18}]
    \label{th: outerfunctionsymmetric}
    For $f:\zone^n \to \zone$ and $g:\zone^n \to \zone$,
    $$\adeg(f \circ g) = \Omega(\adeg(f)\adeg(g)/\log n)$$ when $f:\zone^n \to \zone$ is symmetric. 
\end{theorem}

We also need the following theorems about computing $\Maj_n$ using recursive functions. 
\begin{theorem}[\cite{Goldreich20}]
\label{thm: Goldreich's Maj Constuction}
There exists a constant $C > 0$, such that $\Maj_n\colon\zone^n\to\zone$ is a projection of $\Maj_3^d$ where $d = C\log n$. 
\end{theorem}

\begin{theorem}[\cite{VALIANT1984363}]
\label{thm: Valiant MAJ construction}
There exists a constant $C > 0$, such that $\Maj_n\colon\zone^n\to\zone$ is a projection of $(\AND_2 \circ \OR_2)^d$ where $d = C\log n$. 
\end{theorem}


    

\section{Composition theorem for recursive Majority and alternating $\AND$-$\OR$ trees}
\label{sec:rec-maj}
In this section we give a proof of Theorem~\ref{th: RecMaj composition}. We being with a proof highlight of Lemma~\ref{lem: Maj as amplification function}. Before going into the main theorem we will discuss approximate degree lower bound for the functions recursive Majority and alternating $\AND-\OR$ trees.

\subsection{Approximate Degree of Recursive Majority}
Spectral sensitivity is a nice complexity measure that gives lower bound on $\adeg$. It was used to prove the celebrated conjecture `sensitivity conjecture' by \cite{Huang}. First we will define spectral sensitivity and use it to prove approximate degree lower bound for $\Maj_3^d$ and $(\AND_2 \circ \OR_2)^d$.

We follow the definition from \cite{AaronsonBKRT21} and also state a result from \cite{AaronsonBKRT21} where it was proved that spectral sensitivity lower bounds approximate degree of a function.

\begin{defi}[Spectral Sensitivity]
    \label{defi:spectral sensitivity}
    Let $f:\zone^n \to \zone$ be a Boolean function. The sensitivity graph of $f$ , $G_f = (V, E)$ is a subgraph of the Boolean hypercube,
where $V = {0, 1}^n$ , and $E =\{(x \oplus e_i ) \in V \times V : i \in [n], f (x) \neq f (x \oplus e_i )\}$. That is, $E$ is the set
of edges between neighbors on the hypercube that have different $f$-value. Let $A_f$ be the adjacency
matrix of the graph $G_f$ . We define the spectral sensitivity of $f$ as $\lambda(f) = \|A_f\|.$
\end{defi}

It is well-known that spectral sensitivity is a nice quantity that composes exactly.

\begin{theorem}[\cite{AaronsonBKRT21}]
    For all Boolean function $f$, 
    $\lambda(f) = O(\adeg(f)).$
\end{theorem}

\begin{lemma}
    \label{lem: lambda recMaj}
   $\lambda(\Maj_3^d) = \Theta(2^d)$ where $d$ is the depth of the recursion. 
\end{lemma}

\begin{proof}
   First, we will prove spectral sensitivity of $\Maj_3$. Here is the adjacency matrix of the sensitivity graph of $\Maj_3$. 
\begin{align}
\mathcal{S}=
\begin{pmatrix}
& (000) & (100) & (010) & (001) & (110) & (101) & (011) & (111)\\
(000) & 0 & 0 & 0 & 0 & 0 & 0 & 0 & 0 \\
(100) & 0 & 0 & 0 & 0 & 1 & 1 & 0 & 0 \\
(010)& 0 & 0 & 0 & 0 & 1 & 0 & 1 & 0 \\
(001)& 0 & 0 & 0 & 0 & 0 & 1 & 1 & 0 \\
(110)& 0 & 0 & 0 & 0 & 0 & 0 & 0 & 0 \\
(101)& 0 & 0 & 0 & 0 & 0 & 0 & 0 & 0 \\
(011)& 0 & 0 & 0 & 0 & 0 & 0 & 0 & 0 \\
(111)& 0 & 0 & 0 & 0 & 0 & 0 & 0 & 0 \\
\end{pmatrix}
\end{align}
  Since the adjacency matrix is symmetric, the largest eigenvalue of the matrix is also the norm of the matrix. 
  Note that calculating the eigenvalue of the following $3 \times 3$ matrix is sufficient for our purpose:
\begin{align}
\mathcal{A} = 
     \begin{pmatrix}
      1 & 1 & 0\\
      1 & 0 & 1 \\
      0 & 1 & 1
  \end{pmatrix}
\end{align}
 Eigen value of $A$ is $2$, consequently $\|\mathcal{A}\| = 2$.
 Since spectral sensitivity $(\lambda)$ composes exactly (without any constant overhead) we have $\lambda(\Maj_3^d)=2^d$.
\end{proof}

\begin{corollary}
    \label{lem: approximate degree of recMaj}
   $\adeg(\Maj_3^d) = \Theta(2^d)$ where $d$ is the depth of the recursion.
\end{corollary}
\begin{proof}
From Lemma~\ref{lem: lambda recMaj} it follows that $\adeg(\Maj_3^d) = \Omega(2^d)$ since $\adeg \geq \lambda$ for any Boolean function. 
Also it is known that bounded error quantum query complexity gives upper bound on $\adeg$. \cite{ReichardtS12} showed that bounded error quantum query complexity of $\Maj_3^d$ is $O(2^d)$. Hence, $\adeg(\Maj_3^d)= \Theta(2^d)$.
\end{proof}

\begin{lemma}
\label{lem:lambdaofand-or}
    $\lambda(\AND_2 \circ \OR_2)^d = 2^d $.
\end{lemma}

\begin{proof}
Similar to the $\Maj_3^d$ proof for calculating the largest eigenvalue of the adjacency matrix of $\AND_2$ the following $3 \times 3$ matrix is sufficient for our purpose:
\begin{align}
\mathcal{A} = 
     \begin{pmatrix}
      0 & 0 & 1\\
      0 & 0 & 1 \\
      1 & 1 & 0
  \end{pmatrix}
\end{align}
 Eigen value of $A$ is $\sqrt{2}$, consequently $\|\mathcal{A}\| = \sqrt{2}$.
 Since spectral sensitivity $(\lambda)$ composes exactly (without any constant overhead) we have $\lambda(\AND_2 \circ \OR_2)=2$ and $\lambda(\AND_2 \circ \OR_2)^d=2^d$.
\end{proof}

\begin{corollary}
    \label{lem: approximate degree of and-or formula}
   $\adeg({\AND_2 \circ \OR_2}^d) = \Theta(2^d)$ where $d$ is the depth of the recursion.
\end{corollary}
\begin{proof}
From Lemma~\ref{lem:lambdaofand-or} it follows that $\adeg({\AND_2 \circ \OR_2}^d) = \Omega(2^d)$ since $\adeg \leq \lambda$ for any Boolean function. 
Also it is known from~\cite{KV14} that the upper bound on approximate degree of $({\AND \circ \OR}^d)$ tree is $O(2^d)$. Hence, $\adeg(({\AND \circ \OR}^d))= \Theta(2^d)$.
\end{proof}

\subsection{Proof of Lemma~\ref{lem: Maj as amplification function}}




\MajorityHardnessamplification*

\begin{proof} We will present a proof inspired by the primal-dual view of \cite{She13a}. Fix any constant $0<\varepsilon < 1/2$.
Let $h := f\circ \Maj_t \circ g$ be the composed function, and $p_h : \{0,1\}^{ntm} \to \reals$ be an $\varepsilon$-approximating polynomial for $h$. 

 Further, define $d := \adeg_{\frac{1-\varepsilon}{2}}(g)$. Then, by Lemma~\ref{lem:adeg-dual}, there exists a function $\psi : \zone^m \to \reals$ such that 
\begin{align}
    \sum_{x \in \zone^m}|\psi(x)| & = 1,  \label{eq:proof-psi-l1} \\
    \sum_{x \in \zone^m}\psi(x)\cdot g(x) & > \frac{1-\varepsilon}{2}, \label{eq:proof-psi-corr} 
\end{align}
    and
\begin{align}
    \sum_{x \in \zone^m}\psi(x)\cdot p(x) & = 0 \label{eq:proof-psi-0}
\end{align}
for every polynomial $p$ of degree $< d$. 

Let $\mu$ be the probability distribution on $\zone^m$ given by $\mu(x) = |\psi(x)|$ for $x\in \zone^m$. From \eqref{eq:proof-psi-0}, we have $\sum_{x \in \zone^m}\psi(x) = 0$. Therefore, the sets $\{x \mid \psi(x) < 0\}$ and $\{x \mid \psi(x) > 0\}$ are weighted equally by $\mu$. 
Let $\mu_0$ and $\mu_1$ be the probability distributions obtained by conditioning $\mu$ on the sets $\{x \mid \psi(x) < 0\}$ and $\{x \mid \psi(x) > 0\}$ respectively. Hence,
\begin{align*}
\mu  = \frac{1}{2}\mu_0 + \frac{1}{2}\mu_1, \quad \text{and} \quad \psi = \frac{1}{2}\mu_1 - \frac{1}{2}\mu_0.
\end{align*}
We note an important property of the distributions $\mu_0$ and $\mu_1$ which shows that the error between $\sign(\psi(x))$ and $g(x)$ is low. 
\begin{lemma}
\label{lem:avg-mu-1}
\(\Exp_{x \sim \mu_1}[g(x)] > 1- \varepsilon \).
\end{lemma}
\begin{proof}
\begin{align*}
\Exp_{x \sim \mu_1}[g(x)] 
= \sum_{x \colon \psi(x) >0 } 2\psi(x)\cdot g(x)  
& = \sum_{x\in\{0,1\}^m} 2\psi(x)\cdot g(x) + \sum_{x \colon \psi(x) <0} 2 |\psi(x)|g(x), \\
& > 1 - \varepsilon + \underbrace{\sum_{x \colon \psi(x) <0}2|\psi(x)|g(x)}_{\geq 0} \qquad (\text{from } \eqref{eq:proof-psi-corr})
\end{align*}
\end{proof}
\begin{lemma}
\label{lem:avg-mu-0}
\(\Exp_{x \sim \mu_0}[g(x)] < \varepsilon\).
\end{lemma}
\begin{proof}
\begin{align*}
\Exp_{x \sim \mu_0}[g(x)]  = \sum_{x \colon \psi(x) <0 } 2|\psi(x)|\cdot g(x) & =  \sum_{x \colon \psi(x) >0 } 2\psi(x)\cdot g(x) - \sum_{x\in\{0,1\}^m} 2\psi(x)\cdot g(x), \\
& < \underbrace{\sum_{x \colon \psi(x)>0}2\psi(x)g(x) - 1}_{\leq 0} + \varepsilon < \varepsilon, \qquad (\text{from } \eqref{eq:proof-psi-corr})  
\end{align*}
where the last inequality follows from the fact that $\sum_{x \colon \psi(x)>0}\psi(x)g(x) \leq 1/2$. 
\end{proof}


Consider the following linear operator $L$ that maps functions $h:\zone^{ntm} \to \reals$ to functions $Lh : \zone^n \to \reals$, 
\begin{align}
    \label{defn:linear-operator-maj}
        Lh(z) = \Exp_{\substack{x_{11} \sim \mu_{z_1}\\ x_{12} \sim \mu_{z_1} \\ \vdots \\ x_{1t} \sim \mu_{z_1}}} \Exp_{\substack{x_{21} \sim \mu_{z_2}\\ x_{22} \sim \mu_{z_2} \\ \vdots \\ x_{2t} \sim \mu_{z_2}}}\cdots \Exp_{\substack{x_{n1} \sim \mu_{z_n}\\ x_{n2} \sim \mu_{z_n} \\ \vdots \\ x_{nt} \sim \mu_{z_n}}} [h(x_{11},\ldots , x_{1t},x_{21},\ldots ,x_{2t},\ldots ,x_{n1},\ldots ,x_{nt})]. 
\end{align}
Recall $h=f\circ \Maj_t \circ g$ and $p_h$ be $\varepsilon$-approximating polynomial for $h$. 
Thus by linearity of $L$ we have $\|L(h-p_h)\|_\infty \leq \varepsilon$. We will now observe some useful properties of the linear operator $L$.
\begin{lemma}
\label{lem:L-decreases-deg}
   \(\deg(Lp_h) \leq \deg(p_h)/d\), where $d = \adeg_{\frac{1-\varepsilon}{2}}(g)$. 
\end{lemma}
\begin{proof}
The proof is the same as in \cite{She13a}, but for completeness, we present it here.     
Since $L$ is a linear operator, it suffices to consider the effect of $L$ on monomials of $p_{h}$. A monomial $\mathcal{M}$ of $p_{h}$ is of the form $\prod_{i=1}^np_i(x_{i1},\ldots ,x_{im})$. 
Therefore, \((L\mathcal{M})(z) = \prod_{i=1}^n \Exp_{\mu_{z_i}}[p_i]\). If $p_i$ is of degree $< d = \adeg_{\frac{1-\varepsilon}{2}}(g)$, then $\Exp_{\mu_1}[p_i] = \Exp_{\mu_0}[p_i]$ follows from \eqref{eq:proof-psi-0}. Therefore it doesn't contribute to the degree of $L\mathcal{M}$. 
Otherwise when $p_i$ is of degree $\geq d$, then it contributes at most \emph{one} to the degree of $L\mathcal{M}$. Hence, 
\begin{align*}
\deg(L\mathcal{M}) \leq |\{i \in [n] \mid \deg(p_i) \geq d\}| \leq \frac{\deg(\mathcal{M})}{d}.
\end{align*}
The linearity of $L$ completes the proof. 
\end{proof}
We now show that $Lp_h$ is in fact an approximating polynomial for $f$. 
\begin{lemma}
Fix $0 < \delta <1/2$. Recall $p_h$ is an $\varepsilon$-approximating polynomial for $h = f \circ \Maj_t \circ g$. Let $t = \Theta(\log n + \log(1/\delta))$ where the constant in $\Theta(\cdot)$ depends on $\varepsilon$. Then, $Lp_h$ is a $( \delta + \varepsilon)$-approximating polynomial for $f$. That is, 
    \[\|f - Lp_h \|_\infty \leq \|f - Lh \|_\infty + \|Lh - Lp_h \|_\infty \leq \delta + \varepsilon.\]
\end{lemma}
\begin{proof}
It suffices to show \(\|f - Lh \|_\infty \leq \delta \). To this end, consider $Lh(z)$. 
\begin{align*}
Lh(z) & = \Exp_{\substack{x_{11} \sim \mu_{z_1}\\ x_{12} \sim \mu_{z_1} \\ \vdots \\ x_{1t} \sim \mu_{z_1}}} \Exp_{\substack{x_{21} \sim \mu_{z_2}\\ x_{22} \sim \mu_{z_2} \\ \vdots \\ x_{2t} \sim \mu_{z_2}}}\cdots \Exp_{\substack{x_{n1} \sim \mu_{z_n}\\ x_{n2} \sim \mu_{z_n} \\ \vdots \\ x_{nt} \sim \mu_{z_n}}} [f\circ \Maj_t \circ g (x_{11},\ldots , x_{1t},x_{21},\ldots ,x_{2t},\ldots ,x_{n1},\ldots ,x_{nt})] \\
& = f \left(\Maj_t\left(\Exp_{\mu_{z_1}}[g],\ldots ,\Exp_{\mu_{z_1}}[g]\right),\Maj_t\left(\Exp_{\mu_{z_2}}[g],\ldots ,\Exp_{\mu_{z_2}}[g]\right),\ldots ,\Maj_t\left(\Exp_{\mu_{z_n}}[g],\ldots ,\Exp_{\mu_{z_n}}[g]\right)\right) \\
& = f( z'_1, z'_2, \ldots ,z'_n), 
\end{align*}
where $\| z- z'\|_\infty \leq \delta/n$ because $t=\Theta(\log n + \log(1/\delta))$ and Lemmas~\ref{lem:avg-mu-0} and \ref{lem:avg-mu-1}. 

Therefore, for any $z\in \zone^n$, $|f(z) -Lh(z)| = |f(z) - f(z')| \leq \delta$, since $\|z-z'\|_\infty \leq \delta/n$ and Lemma~\ref{folklorerobust}.
\end{proof} 
Since $Lp_h$ is a $(\delta + \varepsilon)$-approximating polynomial for $f$, we also have \(\deg(Lp_h) \geq \adeg_{\delta + \varepsilon}(f)\). 
We therefore have the following inequalities 
\[ \adeg_{\delta + \varepsilon}(f) \leq \deg(Lp_h) \leq \frac{\deg(p_h)}{\adeg_{\frac{1-\varepsilon}{2}}(g)}. \]
Rewriting we have
\begin{align}
\label{eq:thm-maj}
\adeg_\varepsilon(f\circ \Maj_t\circ g) = \deg(p_h) \geq \adeg_{\delta + \varepsilon}(f) \cdot \adeg_{\frac{1-\varepsilon}{2}}(g).
\end{align}
This completes the proof of Lemma~\ref{lem: Maj as amplification function}. 
\end{proof}

\subsection{Proof of Theorem~\ref{th: RecMaj composition}}

Let $f:\zone^n \to \mathbb{R}$ and $g\colon\zone^m\to\mathbb{R}$ be two functions. We say that $f$ is a \emph{projection} of $g$, denoted $f\leq_{\text{proj}} g$, iff 
\[f(x_1,\ldots ,x_n) = g(a_1,\ldots ,a_m)\] 
for some $a_i \in \zone \cup \{x_1,x_2,\ldots ,x_n\}$. 
That is, $f$ is obtained from $g$ by substitutions of variables by variables of $f$ or constants in $\zone$. 
We note an easy to observe fact about approximate-degree of projections of functions.
\begin{fact}
\label{fact:projections}
    Let $f\colon\zone^n\to\zone$ and $g\colon\zone^m\to\zone$ be such that $f \leq_{\text{proj}} g$, i.e., $f$ is a projection of $g$. Then, \(\adeg(f) \leq \adeg(g)\). 
\end{fact}

Consider the recursive-majority function $\Maj_3^d$ given by the complete $3$-ary tree of height $d$ with internal nodes labeled by $\Maj_3$ and the leaves are labeled by distinct variables. Fix $d \geq C\log\log n$ for a large enough constant $C$. 

Firstly observe that $\Maj_3^d$ is \emph{not} a symmetric function. 
Secondly, it also doesn't have \emph{full} approximate degree (\cite{ReichardtS12}). Thirdly, and finally, its approximate degree is also \emph{not} equal to $\Theta\left(\sqrt{\bs(\Maj_3^d)}\right)$ (see Lemma~\ref{lem: approximate degree of recMaj}, it follows from the fact that $\bs(\Maj_3^d)$ is linear with $\adeg(\Maj_3^d)$). Thus, none of the previous works~\cite{Sherstov12,DavidBGK18,CKMPSS23} imply that approximate degree composes when one of the (inner or outer) functions is recursive-majority $\Maj_3^d$.

\begin{proof}[Proof of Theorem~\ref{th: RecMaj composition}]

Let $\Maj_3^d$ be the recursive-majority function obtained by the complete $3$-ary tree of height $d$ with internal nodes labeled by $\Maj_3$ and the leaves are labeled by distinct variables.  
Let $f\colon\zone^n \to \zone$ be an arbitrary function and consider the approximate degree of the composed function $f\circ \Maj_{t}\circ \Maj_3^d$ where $t=\Theta(\log n)$.   
\begin{align}
    \adeg(f \circ \MAJ_{t} \circ \Maj_3^d) & \leq \adeg(f \circ \MAJ_3^{C \log t} \circ \Maj_3^d) = \adeg(f \circ \Maj_3^d \circ \MAJ_3^{C \log t}) \label{eq: restriction of MAJ adeg ub} \\
    & = O(\adeg(f \circ \Maj_3^d)\cdot\adeg(\MAJ_3^{C \log t}))\label{eq:adeg-composition-ub} \\
    &= O(\adeg(f \circ \Maj_3^d)\cdot\text{poly}(t)). \label{eq: recursive maj ub.}
\end{align}
The first inequality in \eqref{eq: restriction of MAJ adeg ub} follows from the fact that $\Maj_t$ is a projection of $\Maj_3^{C\log t}$ (Theorem~\ref{thm: Goldreich's Maj Constuction}) and Fact~\ref{fact:projections}. Then \eqref{eq:adeg-composition-ub} follows from Theorem~\ref{thm:robust-comp}. 

On the other hand, from Lemma~\ref{lem: Maj as amplification function}, for $t=\Omega(\log n)$ we have 
\begin{align*}
    \adeg(f\circ\Maj_t \circ \Maj_3^d) = \Omega(\adeg(f)\cdot\adeg(\Maj_3^d)).
\end{align*}
Combining with \eqref{eq: recursive maj ub.}, we obtain the lower bound
\begin{align*}
   \adeg(f\circ \Maj_3^d) =  \Omega\left( \frac{\adeg(f) \cdot \adeg(\Maj_3^d)}{\polylog(n)}\right).
\end{align*}
A similar argument shows the following inequalities, where in the last two inequalities we use Theorem~\ref{thm: Valiant MAJ construction} instead of Theorem~\ref{thm: Goldreich's Maj Constuction}, for $d = \Omega(\log n)$,    
\begin{itemize}
\item $\adeg(\Maj_3^d\circ f) = \widetilde\Omega(\adeg(f)\cdot\adeg(\Maj_3^d))$, 
\item \(\adeg( f \circ (\AND_2 \circ \OR_2)^d) = \widetilde\Omega(\adeg(f)\cdot \adeg((\AND_2 \circ \OR_2)^d) )\), and
\item \(\adeg((\AND_2 \circ \OR_2)^d\circ f) = \widetilde\Omega(\adeg(f)\cdot \adeg((\AND_2 \circ \OR_2)^d) )\). 
\end{itemize}
\end{proof}

\section{Composition theorem for recursive functions}
\label{sec:main-thm}
In this section we prove our main theorem (Theorem~\ref{th:for any recursive function}). It shows that the approximate degree composes when either the inner function or the outer function is a recursive function. More formally, 


\anyrecursivefunction*

The following cases of Theorem~\ref{th:for any recursive function} follows from prior works:
\begin{enumerate}
    \item  $f$ or $g$ equals $h^d$ for $h \in\{\PARITY,\neg\PARITY \}$ \cite{Sherstov12}.
    \item $f=h^d$ and $h \in \{\AND,\OR\}$ \cite{DavidBGK18}. 
\end{enumerate}
Therefore, it remains to prove Theorem~\ref{th:for any recursive function} when $h \notin \{\PARITY, \neg\PARITY, \AND, \OR\}$.   
A crucial technical insight that makes the proof work is that when $h \notin \{\PARITY, \neg\PARITY, \AND, \OR\}$ then $\AND_2$ and $\OR_2$ are projections of $h^3$. We can thus simulate $\MAJ$ using a small power of $h$. Thereafter, Lemma~\ref{lem: Maj as amplification function} is used to conclude Theorem~\ref{th:for any recursive function}. We now work out the details. We first state the main technical lemma we need for Theorem~\ref{th:for any recursive function} and then complete the proof of the theorem. Finally, we prove the technical lemma in Section~\ref{sec:technicallemma}.

\begin{restatable}{lemma}{simulateANDOR}\label{thm:simulate_AND-OR}
Let $h:\boolfn{t}$ (where $t\geq 2$) be a Boolean function which depends on all $t$ variables and is not equal to $\PAR/\OR/\AND$. The function $\AND_2$ (and similarly $\OR_2$) can be obtained by setting all but two variables to constants in $h^k$ for $k \leq 3$.
\end{restatable}

We now present the proof of Theorem~\ref{th:for any recursive function} using Lemma~\ref{thm:simulate_AND-OR}.

\begin{proof}[Proof of Theorem~\ref{th:for any recursive function}] 
    Let $h\colon\boolfn{t}$ be such that $h\notin\{\PARITY,\neg\PARITY,\AND,\OR\}$.   
    We know from Lemma~\ref{lem: Maj as amplification function} that 
    \(\adeg(f\circ\Maj_{k}\circ h^d ) = \Omega(\adeg(f) \adeg(h^d))\) where $k=\Theta(\log n)$. Like in the proof of Theorem~\ref{th: RecMaj composition}, we will simulate $\Maj_k$ using $h^\ell$ for sufficiently large $\ell$. 
    From Lemma~\ref{thm:simulate_AND-OR}, it follows that $(\AND_2 \circ \OR_2)^\ell$ is a projection of $h^{6\ell}$. Therefore, we obtain from Theorem~\ref{thm: Valiant MAJ construction} that $\Maj_k$ is a projection of $h^{C\log k}$ for some constant $C>0$. 
    We thus have the following sequence of inequalities, 
    \begin{align*}
        \adeg(f\circ h^d) \geq \adeg(f\circ \Maj_k \circ h^{(d-C\log k)}) & = \Omega(\adeg(f)\adeg(h^{(d-C\log k)})) \\
        & = \Omega\left(\frac{ \adeg(f) \adeg(h^d)}{t^{C\log k}}\right) = \Omega\left(\frac{ \adeg(f) \adeg(h^d)}{\polylog(n)}\right).
    \end{align*}
    Note that the last equality above uses the fact that $t$ is a constant. When $h^d$ is the outer function then we don't need $t$ to be a constant, while the rest of the argument remains the same to give     
    \begin{align*}
        \adeg(h^d \circ g)  = \Omega\left(\frac{\adeg(h^d)\adeg(g)}{\polylog(n)}\right). 
    \end{align*}
\end{proof}
This completes the proof of the main theorem. We now present a proof of Lemma~\ref{thm:simulate_AND-OR}. 

\subsection{Proof of the main technical lemma (Lemma~\ref{thm:simulate_AND-OR})}\label{sec:technicallemma}

We proceed by proving an intermediate result (Lemma~\ref{lem:structure_boolean}) before going to the proof of Lemma~\ref{thm:simulate_AND-OR}.

Suppose we are allowed to \emph{modify} a Boolean function by two operations:
negating some of its variables, and restricting some of the variables to constant values. Lemma~\ref{lem:structure_boolean} proves that almost every Boolean function can be modified to either an $\AND_2$ or an $\OR_2$ function. A restriction of the variables amounts to looking at a smaller hypercube translated to a new point, and negating a variable amounts to rotating the smaller hypercube. In other words, we want to show that there is a \emph{shifted} $\AND_2$ or $\OR_2$ in the Boolean hypercube of $h$ (see Figure~\ref{fig:shifted_OR} for an example).

This shifted $\AND_2/\OR_2$ in the Boolean hypercube of a Boolean function can be concretely defined by the concept of a sensitive block. For a block of variables $S\subseteq [n]$ and an input $x\in \zone^n$, define $x^{\oplus S} \in \zone ^n$ to be the input which flips exactly the variables in S at the input $x$. Given a Boolean function $f:\boolfn{n}$, a block $S$ is called sensitive on $x$ iff $f(x) \neq f(x^{\oplus S})$. A block $S$ is called \emph{minimal sensitive} for $x$ at $f$, if no subset of $S$ is sensitive for $x$ at $f$.

\begin{figure}
    \centering
    \includegraphics[scale=0.1]{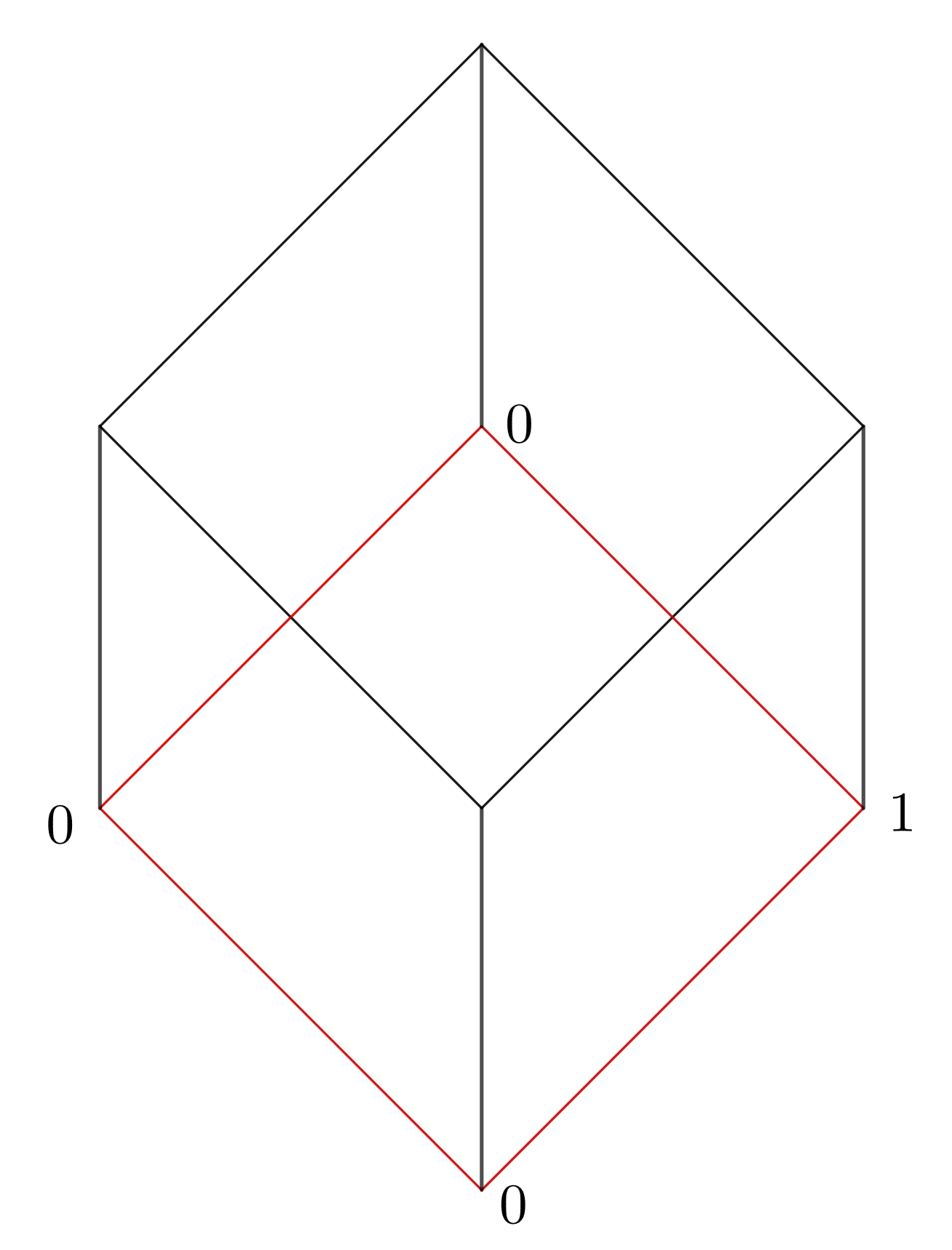}
    \caption{A function on $3$ bits with a shifted $\OR$ marked with red edges.}
    \label{fig:shifted_OR}
\end{figure}


Notice that a shifted $\AND_2/\OR_2$ is a square with three vertices labelled $0$ and one vertex labelled $1$ or vice versa. This gives us a minimal sensitive block on the vertex opposite to the unique value. It can be easily verified that the converse is also true. So, we define a function to have a shifted $\AND_2/\OR_2$ iff it has a minimal sensitive block of size $2$.


We show below that almost all functions have a minimal sensitive block of size $2$.

\begin{lemma} \label{lem:structure_boolean}
Let $h:\boolfn{t}$ (where $t\geq 2$) be a Boolean function which depends on all $t$ variables and is not equal to $\PAR$. Then, there exists an $x\in \zone^t$ such that $h$ has a minimal sensitive block of size $2$ on $x$.
\end{lemma}

\begin{proof} 
We will prove the result using induction on the variables. The statement can be easily verified for $t=2$. 

Define $g_0$ (and $g_1$) to be the restrictions of $h$ by setting $x_t=0$ (and $x_t=1$) respectively. Let $e_y$ be the edge $((y,0),(y,1))$ in the Boolean hypercube, and $S_t := \{e_y: y \in \zone^{t-1}\}$. Color an edge $e_y$ red if $g_0(y) = g_1(y)$, and blue otherwise.

Notice that not all the edges in $S_t$ can be red, otherwise $h$ does not depend on $x_t$. Suppose all the edges in $S_t$ are blue, i.e, $g_1 = \neg g_0$ (in other words, $h = g_0 \oplus x_t$). Since $h$ depends on all variables, then $g_0$ depends on all variables $x_1, x_2, \cdots, x_{t-1}$. If $g_0$ is $\PAR$, then $h$ is also $\PAR$. Implying that $g_0$ is dependent on all its variables and is not $\PAR$. By induction, there exists a minimal sensitive block of size $2$ for $g_0$ (and hence $h$). 

For the rest of the proof, we can assume that there exists both a red and a blue edge in $S_t$.

Let $e_x$ be red and $e_y$ be blue, this means that $g_0(x) = g_1(x)$ but $g_0(y) \neq g_1(y)$. If $x$ and $y$ were at Hamming distance $1$, then vertices $(x,0),(x,1),(y,0)$ and $(y,1)$ will give us the required minimal sensitive block of size $2$.

If $x,y$ are not at Hamming distance $1$, look at any path from $x$ to $y$ in the $t-1$ dimensional hypercube, say $z_0=x, z_1, z_2, \cdots, z_l=y$. The edge $e_{z_0}$ is red and $e_{z_l}$ is blue. Since the color needs to switch at some point, there exist $z_i, z_{i+1}$ at Hamming distance $1$ such that $e_{z_i}$ is red and $e_{z_{i+1}}$ is blue. Again, the vertices $(z_i,0),(z_i,1),(z_{1+1},0)$ and $(z_{i+1},1)$ will give us the required minimal sensitive block of size $2$.

\end{proof}


We are prepared to prove Lemma~\ref{thm:simulate_AND-OR} which shows: given a Boolean function $h$, $\AND_2$ (and $\OR_2$) can be obtained by restricting some of the variables to constants in a very small power of $h$. Compared to Lemma~\ref{lem:structure_boolean}, we need to remove negation and simulate both $\AND_2$ and $\OR_2$ and not just one of them.

We just show how to obtain $\AND_2$, the case for $\OR_2$ is similar. We handle the case of $h$ being monotone and non-monotone separately. 
\paragraph*{Monotone $h$:}
This case is simpler, and $\AND_2$ can be obtained as a restriction of $h$ itself. Let a minimal $1$-input be a $x\in \zone^t$ such that setting any $1$ bit of $x$ to $0$ changes the value of $h$. If there is a minimal $1$-input $x$ of Hamming weight more than $2$, we get a $\AND_2$ by choosing any two indices which are $1$ in $x$. The following claim finishes the proof for monotone functions.
\begin{claim}
Let $h:\boolfn{t}$ be a monotone Boolean function which depends on all variables. If there is no minimal $1$-input with Hamming weight more than $2$, then $h$ is the $\OR$ function.      
\end{claim}
\begin{proof}
By abusing the notation, let $0$ denote the all $0$ input.
Since the function is monotone but not constant, we know that $h(0)= 0$. Let $S \subseteq [t]$ capture the indices such that the corresponding Hamming weight $1$-input has function value $0$,
\[ S = \{i: h(0^{\oplus i}) = 0 \} .\]

For a $y\in \zone^t$, if the set of $1$-indices are not a subset of $S$, then $h(y) = 1$ by monotonicity. If the set of $1$-indices are a subset of $S$, then $h(y) = 0$ because there is no minimal $1$-input with Hamming weight more than $2$. 

In other words, $h$ is the $\OR$ function on the remaining $[t]-S$ variables. Since $h$ depends on all the $t$ variables, $ h$ is the $\OR$ function. 

\end{proof}
\begin{figure}
    \centering
    \includegraphics[scale=0.5]{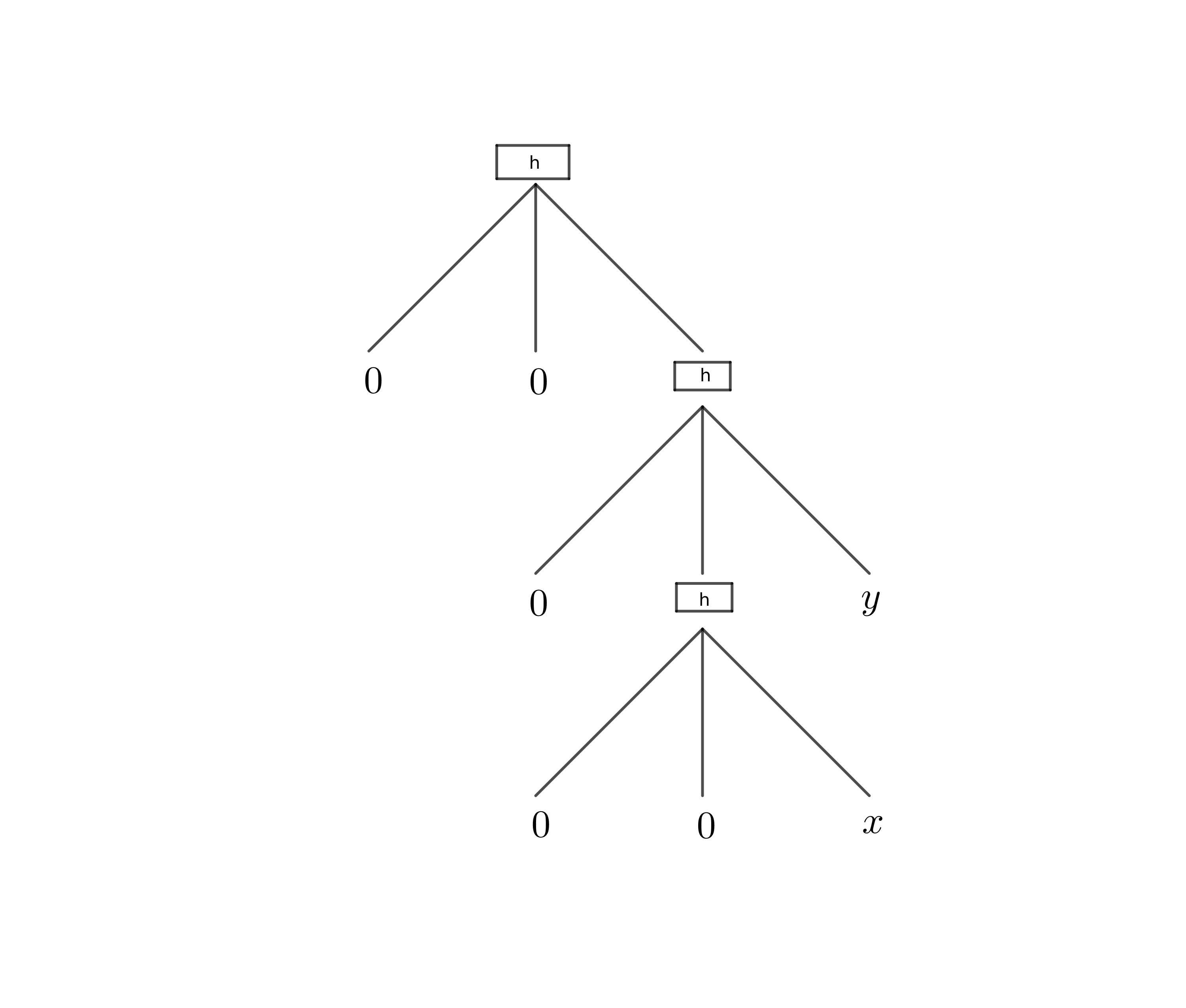}
    \caption{An example for constructing $\AND_2$ using a non-monotone function. Let $\mathsf{h: \zone^3 \to \zone}$ be $0$ at $x= 001$ and $1$ otherwise. We use the shifted $\OR_2$/minimal sensitive block at $001$ with indices $\{2,3\}$.}
    \label{fig:embedding and}
\end{figure}
\paragraph*{Non-monotone $h$:}
Since $h$ is a non monotone function, there exists an input $a \in \zone^t$ and an index $i \in [t]$ such that $h(a) = 1$, $a_i = 0$ and $h(a^{\oplus i}) = 0$. Restricting the variables according to $a$ (except the $i$-th bit) gives $h_1(x_i) = \neg x_i$. 

From Lemma~\ref{lem:structure_boolean}, there exists a $b \in \zone^t$ such that $h$ has a minimal sensitive block of size $2$ on $b$ (shifted $\AND_2$/$\OR_2$). The main idea of this proof is to use negation and this shifted $\AND_2$/$\OR_2$ (Figure~\ref{fig:embedding and} gives an example). 

For the formal proof, without loss of generality assume that the block have indices $1,2$ (that means $h(b) = h(b^{\oplus \{1\}}) = h(b^{\oplus \{2\}}) \neq h(b^{\oplus \{1,2\}})$). We will finish the proof by considering the two cases $h(b)=0$ and $h(b)=1$. 

\begin{itemize}
    \item $h(b)=0$ (shifted $\AND_2$): Suppose $b_1=0$ and $b_2=1$ (other cases can be handled similarly). Notice that $\AND_2(x,y) = h(x, \neg y, b_3,\cdots, b_t)$, giving us $\AND_2(x,y) = h(x, h_1(y), b_3,\cdots, b_t)$.
    \item $h(b)=1$ (shifted $\OR_2$): Suppose $b_1=1$ and $b_2=0$ (other cases can be handled similarly). Notice that $\OR_2(x,y) = h(x, \neg y, b_3,\cdots, b_t)$; using De Morgan's law,
    \[ \AND_2(x,y) = \neg \OR_2(\neg x, \neg y) = \neg h(\neg x, y, b_3, \cdots, b_t) = h_1(h(h_1(x), y, b_3, \cdots, b_t)) \]
\end{itemize}

Since $h_1$ is also a restriction of $h$, the proof is complete. 

\section{Functions with full sign degree as hardness amplifier}
\label{appendix:composition theorem for hardness amplifier}




The aim of this section is to show that any Boolean function with full sign degree can be used as a hardness amplifier, similar to $\Maj_t$ in Lemma~\ref{lem: Maj as amplification function}.

\begin{lemma}
     \label{lem: full degree function as hardness amplifier function }
    For any Boolean function $f:\zone^n \to \zone$ and $g:\zone^m \to \zone$ we have,
    \begin{align*}
        \adeg(f \circ \amp_t \circ g) = \Omega(\adeg(f) \cdot \adeg(\amp_t) \cdot \adeg(g))
    \end{align*}
    where $\amp_t$ is any Boolean function on $t$ bits with full sign degree. 
\end{lemma}



For the proof of lemma~\ref{lem: full degree function as hardness amplifier function } we will be using the following composition theorem by \cite{Sherstov11} for the classes of outer function with full approximate degree.

\begin{theorem}[\cite{Sherstov11}]
\label{th:fulladegcompositionsherstov}
    For any Boolean function $h:\zone^n \to \zone$ and $g:\zone^n \to \zone$, $$\adeg_{\frac{1-\varepsilon^{-n}}{2}}( h\circ g) = \Omega(\adeg(h)\cdot \adeg(g))$$ where $h$ is a function with full sign-degree. 
\end{theorem}

 
The proof will be along with the same line of Lemma~\ref{lem: Maj as amplification function}. Instead of using the dual of the inner function $g$, we will be using the dual of $\amp_t \circ g$. 
Fix any constant $0<\varepsilon <1$ and consider ${\frac{1-\varepsilon^{-t}}{2}}$-approximating polynomial of $\amp_t \circ g$. Let $\adeg_{\frac{1-\varepsilon^{-t}}{2}}(\amp_t \circ g) =: d$. By Lemma~\ref{lem:adeg-dual}, there exists a function $\psi : \zone^{mt} \to \reals$ such that 
\begin{align}
    \sum_{x \in \zone^{mt}}|\psi(x)| & = 1, \\ \label{eq:ampproof-psi-l1} \\
    \sum_{x \in \zone^{mt}}\psi(x)\cdot g(x) & > \frac{1-\varepsilon^{-t}}{2}, \label{eq:ampproof-psi-corr} 
\end{align}
    and
\begin{align}
    \sum_{x \in \zone^{mt}}\psi(x)\cdot p(x) & = 0 \label{eq:ampproof-psi-0}
\end{align}
    for every polynomial $p$ of degree $< d$. 

    Similar to Lemma~\ref{lem: Maj as amplification function} let $\mu$ be the probability distribution on $\zone^{mt}$ given by $\mu(x) = |\psi(x)|$ for $x\in \zone^{mt}$. Define $\mu_0$ and $\mu_1$ similarly. We have the following properties of $\mu_0$ and $\mu_1$.

     \begin{lemma}
    \label{lem:ampavg-mu-1}
        \begin{align}
        \label{eq:ampavg-mu-1}
            \Exp_{x \sim \mu_1}[\amp_t \circ g(x)] > 1- \varepsilon^{-t},
        \end{align} and 
        \begin{align}
        \label{eq:ampavg-mu-0}
            \Exp_{x \sim \mu_0}[\amp_t \circ g(x)] < \varepsilon^{-t}.
        \end{align}
    \end{lemma}


Consider the following linear operator $L$ that maps functions $h:\zone^{ntm} \to \reals$ to functions $Lh : \zone^n \to \reals$, 

 \begin{align}
    \label{defn:linear-operator}
        Lh(z) = \Exp_{x_1 \sim \mu_{z_1}}\Exp_{x_2 \sim \mu_{z_2}}\cdots \Exp_{x_n \sim \mu_{z_n}} [h(x_1,x_2,\ldots ,x_n)]. 
    \end{align}

    Let $H:= f\circ \amp_t \circ g $, where $f : \zone^n \to \zone$ and $g:\zone^m \to \zone$. Further, let  $p_H : \{0,1\}^{ntm} \to \reals$ be an $\varepsilon$-approximating polynomial for the composed function $H : \zone^{ntm}\to \zone$. That is, $\| H - p_H\|_{\infty} \leq \varepsilon$. Then, it follows from the linearity that 
    \[\|L(H-p_H)\|_{\infty} \leq \varepsilon. \]

 We now explore the behavior of $L$ on $H$ and $p_{H}$. 
    \begin{lemma}
    \label{lem:L-reduces-degH}
        \begin{align*}
          \deg(Lp_{H}) \leq \frac{\deg(p_{H})}{\adeg_{\frac{1-\varepsilon^{-t}}{2}}(g)} .
        \end{align*}
    \end{lemma}
   \begin{proof}
       Follows exactly as Lemma~\ref{lem:L-decreases-deg}.
   \end{proof}
    Now it is left to prove that $Lp_H$ is an approximating polynomial for $f$.

\begin{lemma}
\label{lem:mainapproximatngLH}
Fix $0 < \delta <1$. Recall $p_H$ is an $\varepsilon$-approximating polynomial for $H = f \circ \amp_t \circ g$. Let $t = \Omega(\log n + \log(1/\delta))$ where the constant in $\Theta(\cdot)$ depends on $\varepsilon$. Then, $Lp_H$ is a $( \delta + \varepsilon)$-approximating polynomial for $f$. That is, 
    \[\|f - Lp_H \|_\infty \leq \|f - LH \|_\infty + \|LH - Lp_H \|_\infty \leq \delta + \varepsilon.\]
\end{lemma}



\begin{proof}
    It suffices to bound \(\|f - Lh \|_\infty\). To this end, consider $Lh(z)$. 

      \begin{align}
            (Lh)(z) & = \Exp_{x_1 \sim \mu_{z_1}}\Exp_{x_2 \sim \mu_{z_2}}\cdots \Exp_{x_n \sim \mu_{z_n}} [h(x_1,x_2,\ldots ,x_n)] \\
            & = \Exp_{x_1 \sim \mu_{z_1}}\Exp_{x_2 \sim \mu_{z_2}}\cdots \Exp_{x_n \sim \mu_{z_n}} [f\circ \amp_t \circ g (x_1,x_2,\ldots ,x_n)] \\
            & {=} f\left(\Exp_{\mu_{z_1}}[\amp_t \circ g], \Exp_{\mu_{z_2}}[\amp_t \circ g], \ldots , \Exp_{\mu_{z_n}}[\amp_t \circ g]\right) \\
            & {=} f\left(z_1', \dots, z_n' \right)
        \end{align} where $\| z- z'\|_\infty = O(\delta/n)$ because $t=\Omega(\log n + \log(\frac{1}{\delta}))$ and Lemmas~\ref{lem:ampavg-mu-1}. 

    Therefore, for arbitrary $z\in \zone^n$, $|f(z) -Lh(z)| = |f(z) - f(z')| \leq \varepsilon_1$, where $\|z-z'\|_\infty = O(1/n)$. 

Since $Lp_h$ is an $(\varepsilon_1 + \varepsilon)$-approximating polynomial for $f$, we also have \[\deg(Lp_h) \geq \adeg_{\varepsilon_1 + \varepsilon}(f).\]
    We therefore have the following inequalities 
    \[ \adeg_{\varepsilon_1 + \varepsilon}(f) \leq \deg(Lp) \leq \frac{\deg(p)}{\adeg_{\frac{1-\varepsilon^{-t}}{2}}(\amp_t \circ g)}. \]
    This implies \[ \deg(p) \geq \adeg_{\varepsilon_1 + \varepsilon}(f) \cdot {\adeg_{\frac{1-\varepsilon^{-t}}{2}}(\amp_t \circ g)} \geq \adeg_{\varepsilon_1 + \varepsilon}(f) \cdot \adeg (\amp_t)\cdot \adeg(g),\]
    where the last inequality follows from Theorem~\ref{th:fulladegcompositionsherstov} by \cite{Sherstov11}.
    \begin{align}
    \label{eq:thm-xor}
    \adeg_\varepsilon(f\circ \amp_t\circ g) = \deg(p) \geq  \adeg_{\varepsilon_1 + \varepsilon}(f) \cdot \adeg(\amp_t) \cdot \adeg_{1/3}(g)
    \end{align}
   
\end{proof}

 From the proof technique, it is clear that if we start with an $({\frac{1-\frac{\delta}{n}}{2}})$-approximating polynomial for the inner function $g$ then we are fine with the error accumulated. So, the following known results can be derived as a corollary of the above Lemma. 
 
\begin{corollary}[\cite{Sherstov13}]
    For any Boolean function $f:\zone^n \to \zone$ and $g: \zone^m \to \zone$ we have $\adeg_{\varepsilon}(f \circ g) = \Omega(\adeg_{\epsilon_1 + \varepsilon}(f)\adeg_{\frac{1-\frac{\delta}{n}}{2}}(g)),$ where $\epsilon, \varepsilon$ is some constant in $(0,\frac{1}{2})$.
\end{corollary}

From above it also follows 

\begin{corollary}[\cite{Sherstov13}]
     For any Boolean function $f:\zone^n \to \zone$ and $g: \zone^m \to \zone$ we have $\adeg_{\varepsilon}(f \circ g) = \Omega(\adeg(f)\adeg_{\pm}(g)),$ where $\adeg_{\pm}(g)$ denotes the sign-degree of $g$.
\end{corollary}
\subsection{Composition theorem for recursive functions with full sign degree}


In this section, we will prove the composition theorem for a few more recursive functions. To prove our theorem we will be using Lemma~\ref{lem: full degree function as hardness amplifier function } which is a generalization of Lemma~\ref{lem: Maj as amplification function}. For the sake of completeness, we give the proof of Lemma~\ref{lem: full degree function as hardness amplifier function } in Appendix~\ref{appendix:composition theorem for hardness amplifier}. 

\cite{Saks_1993} and \cite{Anthony95} shows that,

\begin{theorem}(\cite{Saks_1993} and \cite{Anthony95})
    \label{th: full sign degree}
    Almost all the function $f: \zone^n \to \zone$ sign degree is high,  $\adeg_{\pm}(f)= \Omega(n)$.
\end{theorem}

   So, Theorem~\ref{th: recursive function composition} gives a composition theorem for the recursive version of all such functions being inner or outer functions. Note that when the outer function has a full sign degree applying \cite{Sherstov12} we get $\adeg(\amp_t)$ in the lower bound part in Lemma~\ref{lem: full degree function as hardness amplifier function }, using which we can prove composition theorem for some classes of functions where we are not loosing the $\mathsf{polylog(n)}$ in the lower bound part compared to Theorem~\ref{th:for any recursive function}.


 \begin{theorem}
 \label{th: recursive function composition}
 For any Boolean function $f:\zone^n \to \zone$ and $g:\zone^n \to \zone$ the following holds:
$\adeg(f \circ g) = {\Omega}(\adeg(f)\adeg(g))$ when, 
\begin{itemize}
    \item $f=\amp_t^d$ and $g$ is any Boolean function,
    \item $g=\amp_t^d$ and $f$ is any Boolean function.
     \end{itemize}
    where $\amp_t$ is any Boolean function with full sign degree and $t = \Omega(\log n)$.
 \end{theorem}

The proof is completely similar to the proof of Theorem~\ref{th: RecMaj composition}, we are giving the proof for the sake of completeness. Note that here we are not losing the $\log n$ compared to Theorem~\ref{th: RecMaj composition}. 

\begin{proof}
    From lemma~\ref{lem: full degree function as hardness amplifier function }, we have $\adeg(f \circ \amp_t \circ g) = \Omega(\adeg(f)\cdot \adeg(\amp_t) \cdot  \adeg(g))$.
    
    If $f=\amp_t^k$,
    \begin{align}
        \adeg(\amp_t^k \circ \amp_t \circ g) \geq \adeg(\amp_t^k)\cdot \adeg(\amp_t) \cdot \adeg(g)
        \label{eq:amp as outer function}
    \end{align}.

    On the other hand using associativity if composition we have,
    \begin{align}
        \adeg(\amp_t \circ \amp_t^k \circ g) \geq \adeg(\amp_t^k \circ \amp_t \circ g) \label{eq:amp associativity}
    \end{align}
     Applying Theorem~\ref{thm:robust-comp} we also have the following,
     \begin{align}
         \adeg(\amp_t \circ \amp_t^k \circ g) \leq \adeg(\amp_t)\cdot \adeg(\amp_t^k \circ g) \leq \adeg(\amp_t) \cdot \adeg(\amp_t^k \circ g).
     \end{align}
     From equation~\eqref{eq:amp as outer function}, \eqref{eq:amp associativity} it follows that,
     \begin{align*}
         \adeg(\amp_t^k \circ g) \geq \adeg(\amp_t^k) \cdot \adeg(g).
     \end{align*}

    For $g = \Maj_t^k$, we can derive the proof in a similar fashion. 
\end{proof}

\section{Conclusion}
Towards the main open problem of approximate degree composition, we have the following immediate question in light of Lemma~\ref{lem: Maj as amplification function}. 
Can we upper bound $\adeg(f \circ \Maj_t \circ g)$ in terms of $\adeg(f \circ g)$? Precisely,
\begin{open question}
Is $\adeg(f \circ \Maj_t \circ g) = \Tilde{O}(\adeg(f \circ g))$, where $t = \Theta(\log n)$ and $n$ is the arity of the outer function $f$? 
\end{open question}
Observe that an affirmative solution to the above question solves the composition question for approximate degree in positive. 
Another interesting question is to find other classes of functions for which the analogue of Lemma~\ref{lem: full degree function as hardness amplifier function } holds.  
\begin{open question}
 Find non-trivial classes of functions $H$ such that $\adeg(f \circ h \circ g) = \Tilde{\Omega} (\adeg(f)\cdot \adeg(h) \cdot \adeg(g) )$ for all $h\in H$?
\end{open question}

It has the following two useful implications. First, this gives composition for functions $h \in H$. 
In particular, when one of the functions $h$ (inner or outer) belongs to the class $H$ then 
\(\adeg(f \circ h \circ g) =  \Tilde\Omega(\adeg(f)\cdot \adeg(h) \cdot \adeg(g)) \) along with Theorem~\ref{thm:robust-comp} implies 
\[ \adeg(h \circ g) = \Tilde\Omega(\adeg(h) \cdot \adeg(g)) \quad \text{ and } \quad \adeg(f \circ h) = \Tilde\Omega(\adeg(f) \cdot \adeg(h)).\]
Second, a function $h \in H$ 
can be used as `hardness amplifier' functions. 

Another very interesting question that may provide us insights to make progress towards the main question of approximate degree composition is to prove that approximate degree composes when the inner function is $\OR$.
\begin{open question}
Show that $\adeg(f \circ \OR) = \Tilde{\Omega}(\adeg(f).\adeg(\OR))$. 
\end{open question}


\bibliography{ref}
\newpage
\appendix


\section{One-sided approximate degree}

Here is some analogous result to Lemma~\ref{lem: Maj as amplification function} in the case of one-sided approximate degree. Let us define the concept of one-sided approximate degree first.

\begin{defi}
    A polynomial $p$ is a one-sided  $\varepsilon$-approximation to $f$ if 
    \begin{enumerate}
        \item for all $x \in f^{-1}(1)$, $|p(x)-1| \leq \varepsilon$, and 
        \item for all $x \in f^{-1}(0)$, $p(x) \leq \varepsilon$. 
    \end{enumerate}
\end{defi}
\begin{defi}
    The one-sided $\varepsilon$-approximate degree of $f$, denoted $\odeg_\varepsilon(f)$, is the least degree of a real polynomial $p$ that is a one-sided $\varepsilon$-approximation to $f$. 
\end{defi}
\begin{lemma}[\cite{BunT22}]
\label{lem:odeg-char}
Let $f\colon\zone^n \to \zone$ be a Boolean function. Then $\odeg_\varepsilon(f) > d$ if and only if there exists a function $\psi\colon\zone^n\to\reals$ such that 
\begin{align*}
    \sum_{x \in \zone^n} |\psi(x)| & = 1, \\
    \sum_{x \in \zone^n}\psi(x)f(x) & > \varepsilon, \\
    \sum_{x \in \zone^n}\psi(x)p(x) & = 0 \text{ for every polynomial $p$ of degree }\leq d, \quad \text{and} \\
    \psi(x) & \leq 0 \text{ for all } x \in f^{-1}(0). 
\end{align*}
\end{lemma}

 The dual formulation for one-sided approximate degree is the following:

\begin{corollary}
\label{cor: onesided amp}
    For any Boolean function $f:\zone^n \to \zone$ we have the following:
    \begin{align*}
        \adeg(f\circ \AND_t \circ \OR) = \Omega( \adeg(f) \cdot \adeg(\OR))
    \end{align*} for $t=\Omega(\log n)$.
\end{corollary}
\begin{proof}
     When $g=\OR$, then we can use \emph{one-sidedness} of dual of $\OR$, i.e., $\Exp_{\mu_1}[\OR] = 1$ and $\Exp_{\mu_0}[\OR] < \varepsilon$, to amplify using $\AND$ instead of $\Maj$. Which gives  $\adeg(f\circ \AND_t \circ \OR) = \Omega( \adeg(f) \cdot \adeg(\OR))$. 
\end{proof}
Precisely for one sided dual,
$\Exp_{\mu_1}[g] = 1$ and $\Exp_{\mu_0}[g] < \varepsilon$ which gives the following observation.

\begin{observation}
\label{obs: and amplification}
For any Boolean function $f:\zone^n \to \zone$ and $g:\zone^m \to \zone$ we have the following:
     $\adeg(f\circ \AND_t \circ g) = \Omega( \adeg(f) \cdot \odeg(g))$ for $t=\Omega(\log n)$.
\end{observation}


\end{document}